\documentclass[aps,prd,superscriptaddress,nofootinbib,preprintnumbers,floatfix]{revtex4-2}

\usepackage{amsmath,amssymb,amsfonts}
\usepackage{graphicx}
\usepackage{hyperref}
\usepackage{color}
\usepackage{bm}
\usepackage{microtype}
\usepackage{mathrsfs}
\usepackage{amsthm}
\newtheorem{proposition}{Proposition}

\hypersetup{
    colorlinks=true,
    linkcolor=blue,
    citecolor=blue,
    urlcolor=blue
}

\begin{document}

\title{Fold-Point Scaling and Phenomenological Cascade Closures in Regular Black Hole Spacetimes}

\author{Hoang Van Quyet}
\affiliation{Department of Physics, Hanoi Pedagogical University 2, Xuan Hoa, Phu Tho, Vietnam}
\email{hoangvanquyet@hpu2.edu.vn}

\date{\today}

\begin{abstract}
Semiclassical analyses of regular black holes suggest that inner-horizon instabilities can drive the trapped region to evaporate faster than the Hawking time and seed a transient anti-trapped region. Whether this exchange repeats into a cascade terminating in a horizon-free configuration is an open question requiring a self-consistent solution of the semiclassical Einstein equations, which we do not attempt. Instead, within the static Bardeen family, we show that the renormalized-stress-tensor (RSET) result of Arrechea, Liberati and Spadafora [arXiv:2608.03538] does \emph{not} by itself force any damping of a hypothetical cascade near extremality, since its near-extremal limit at fixed cycle duration is finite and nonzero. We then prove, under explicit nondegeneracy hypotheses on any static metric family whose horizons merge at a fold point, that the inner-horizon surface gravity scales as $|\kappa_-|\propto(M-M_{\rm crit})^{1/2}$ with an explicit remainder, and verify this in two regular black hole families, Bardeen and Hayward, with a quantified fit-window/grid systematic envelope in addition to the regression error. A phenomenological closure built on this scaling gives a mass sequence whose gap decays as a power law, $x_n\sim n^{-2}$, generalizing to $x_n\sim n^{-1/(pq)}$ for a fold exponent $p$ and an uncalibrated closure exponent $q$; the cycle count $N(\varepsilon)$ is an explicit property of this closure, not a physical prediction, and we track a fractional per-cycle update diagnostic (not a physical adiabaticity test) to show where the discrete recursion changes slowly. As a separate, complementary exercise, we construct an effective adiabatic model for the outer-horizon mass by matching the exact Einstein tensor of a generalized Vaidya-Bardeen ansatz to the cited RSET flux; the resulting mass-loss law includes an $O(1)$ correction factor absent from the naive Schwarzschild-Vaidya relation, relaxes to $M_{\rm crit}$ asymptotically rather than in finite time, and yields a derived timescale that is shorter than the inner-horizon amplification timescale away from extremality by up to four orders of magnitude but longer below an explicitly computed crossover very close to extremality. We emphasize that this outer-horizon model rests on a specific reduction and adiabatic ansatz, not on a first-principles non-adiabatic solution, and we distinguish throughout what is derived under stated assumptions from what the closure simply postulates.
\end{abstract}

\maketitle

\section{Introduction}
\label{sec:intro}

Regular black hole spacetimes, first proposed by Bardeen~\cite{Bardeen1968} and generalized by Hayward~\cite{Hayward2006}, replace the classical curvature singularity of gravitational collapse with a smooth core, at the price of an inner (Cauchy) horizon $r_-$ in addition to the usual outer horizon $r_+$. In semiclassical gravity, quantized matter fields on such backgrounds source a slow, thermal evaporation of the outer horizon at the Hawking temperature~\cite{Hawking1975}, on the parametrically long timescale $O(M^3)$. A spacelike 2-surface is called trapped, anti-trapped, or untrapped according to the sign of the outgoing and ingoing null expansions $\theta_+,\theta_-$: trapped if both are negative, anti-trapped if both are positive, and untrapped otherwise~\cite{CarballoRubio2021}. The region between $r_-$ and $r_+$ of a black hole is trapped; the corresponding region of a white hole is anti-trapped.

Inner horizons are unstable both classically, through mass inflation~\cite{SimpsonPenrose1973,PoissonIsrael1990}, and semiclassically, through the exponential blueshift amplification of vacuum fluxes~\cite{CarballoRubio2021,DiFilippo2022,CarballoRubio2024,Barcelo2022}. Using the two-dimensional renormalized stress-energy tensor (RSET) formalism~\cite{ChristensenFulling1977,DaviesFulling1977,Polyakov1981,BirrellDavies1982,FabbriNavarroSalas2005} in the $s$-wave (Polyakov) approximation, and building on the collapse models of Refs.~\cite{CarballoRubio2008,BarenboimFrolovKunstatter2024,BarceloCarballoRubio2016}, which showed numerically that this inside-out evaporation of the trapped region proceeds much faster than $O(M^3)$, Arrechea, Liberati and Spadafora~\cite{Arrechea2026} (henceforth ALS) recently gave an analytic, fixed-background treatment of a single black-hole-to-white-hole transition: the outgoing RSET flux inside a trapped region is exponentially amplified near the inner horizon, and its backreaction, diagnosed through the Raychaudhuri equation, can flip the sign of the ingoing expansion, converting the trapped region into an anti-trapped one. Independently, Boyanov, Hilditch and Semi\~ao~\cite{BoyanovHilditchSemiao2026} obtained a self-consistent numerical solution of the semiclassical Einstein equations for a single such transition, confirming the qualitative picture. ALS explicitly note that the resulting anti-trapped region could in turn seed a further trapped region, opening the possibility of a cascade, but they are equally explicit that establishing the multi-cycle evolution requires a self-consistent backreaction calculation that is beyond the scope of their fixed-background diagnostic, and that the causal structure could in principle become more complicated than a single areal-radius function can describe (e.g., through wormhole necks).

This scenario is conceptually distinct from black-to-white-hole transitions proposed in loop quantum gravity, where the transition is a genuinely quantum-gravitational tunneling event through a Planck-curvature region~\cite{HaggardRovelli2015,Bianchi2018,HanRovelliSoltani2023}; here, by contrast, the transition is driven by matter-field vacuum polarization on a fixed low-curvature background, with no Planckian physics invoked. We mention the LQG scenario only to distinguish the mechanisms, not to suggest either supports the other.

The present paper does not perform the self-consistent multi-cycle calculation called for in Ref.~\cite{Arrechea2026}. Instead, it makes three more modest contributions. First, we show explicitly that the ALS flux formula does \emph{not} by itself imply that a hypothetical cascade damps near extremality; any such damping must be imposed as a modeling assumption. Second, we prove under explicit local hypotheses, and verify in two independent regular black hole families, that the inner-horizon surface gravity vanishes as a square root of the mass gap at a fold-type horizon merger, and work out what this implies for a minimal closure built on it, including a full classification in terms of the fold exponent $p$ and an uncalibrated closure exponent $q$. Third, for the specific sub-problem of outer-horizon evaporation within one trapped phase, we construct an effective adiabatic model from the exact Einstein tensor of a generalized Vaidya-Bardeen ansatz matched to an ALS formula, and show that even this explicit construction does not by itself reproduce the finite-time trapped-region disappearance seen in non-adiabatic numerical solutions, which we take as an illustration of where adiabatic reasoning about this problem is insufficient. To be unambiguous about what follows: this is a conditional, model-internal analysis, not a proposed physical evolution law, and we return to this distinction throughout.

Section~\ref{sec:geometry} reviews the two regular black hole backgrounds used. Section~\ref{sec:rset} summarizes the ALS RSET result and the non-vanishing of its near-extremal limit at fixed cycle duration. Section~\ref{sec:universality} states and proves the fold-scaling proposition and verifies it in both metrics. Section~\ref{sec:cascade} states the closure explicitly and derives the general $(p,q)$ classification. Section~\ref{sec:numerics} gives the numerical protocol and results, including sensitivity and a fractional-update diagnostic. Section~\ref{sec:vaidya} constructs an effective adiabatic model for the outer-horizon mass and compares its timescale to the inner-horizon one. Section~\ref{sec:phenomenology} discusses a purely illustrative radiation profile. Section~\ref{sec:conclusions} concludes.

\section{Regular Black Hole Backgrounds}
\label{sec:geometry}

We use the static, spherically symmetric line element in ingoing Eddington--Finkelstein coordinates,
\begin{equation}
\label{eq:metric}
ds^2 = -f(r)\, dv^2 + 2\, dv\, dr + r^2\, d\Omega^2,
\end{equation}
with $v$ the advanced time and $d\Omega^2=d\theta^2+\sin^2\theta\, d\phi^2$. Our primary background is the Bardeen metric function~\cite{Bardeen1968},
\begin{equation}
\label{eq:f_bardeen}
f_{\rm B}(r) = 1 - \frac{2Mr^2}{(r^2+g^2)^{3/2}},
\end{equation}
with $M$ the ADM mass; $g$ sets the curvature scale of the regular core and is tied to an effective magnetic charge in the metric's nonlinear-electrodynamics realization~\cite{Bardeen1968,Ayon-BeatoGarcia2000}, though we treat it as a free phenomenological length, as is standard in this literature. As a second, independent check of genericity (Sec.~\ref{sec:universality}), we also use the Hayward metric function~\cite{Hayward2006},
\begin{equation}
\label{eq:f_hayward}
f_{\rm H}(r) = 1 - \frac{2Mr^2}{r^3+2Ml^2},
\end{equation}
with core scale $l$. Both solve $f(r)=0$ for two positive roots $r_+>r_-$ when the core scale is small enough, merging at a critical mass:
\begin{align}
\label{eq:Mcrit_B}
\text{Bardeen:}&\quad M_{\rm crit} = \frac{3\sqrt3}{4}\, g, \quad r_0=\sqrt2\, g,\\
\label{eq:Mcrit_H}
\text{Hayward:}&\quad M_{\rm crit} = \frac{3\sqrt3}{4}\, l, \quad r_0=\sqrt3\, l.
\end{align}
For $M<M_{\rm crit}$, $f(r)>0$ everywhere and the geometry is horizon-free, in both cases.

The surface gravity $\kappa(r)\equiv\tfrac12 f'(r)$ for the Bardeen metric is
\begin{equation}
\label{eq:kappa}
\kappa_{\rm B}(r) = \frac{M\,r\,(r^2-2g^2)}{(r^2+g^2)^{5/2}}.
\end{equation}
At $r=r_+$ ($r_+^2>2g^2$) this is positive; at $r=r_-$ ($r_-^2<2g^2$) it is negative. The sign difference reflects the different orientation of the horizon-crossing null congruence at $r_\pm$ and is consistent with the well-known blueshift instability of inner horizons~\cite{PoissonIsrael1990}; the sign alone is not, by itself, a proof of a ``repulsive force,'' which would be a separate dynamical statement about test-particle motion that we do not make. Both $\kappa_+$ and $\kappa_-$ vanish as $M\to M_{\rm crit}$; Sec.~\ref{sec:universality} shows this vanishing is generic and quantifies its rate.

\begin{figure}[t]
\centering
\includegraphics[width=0.8\textwidth]{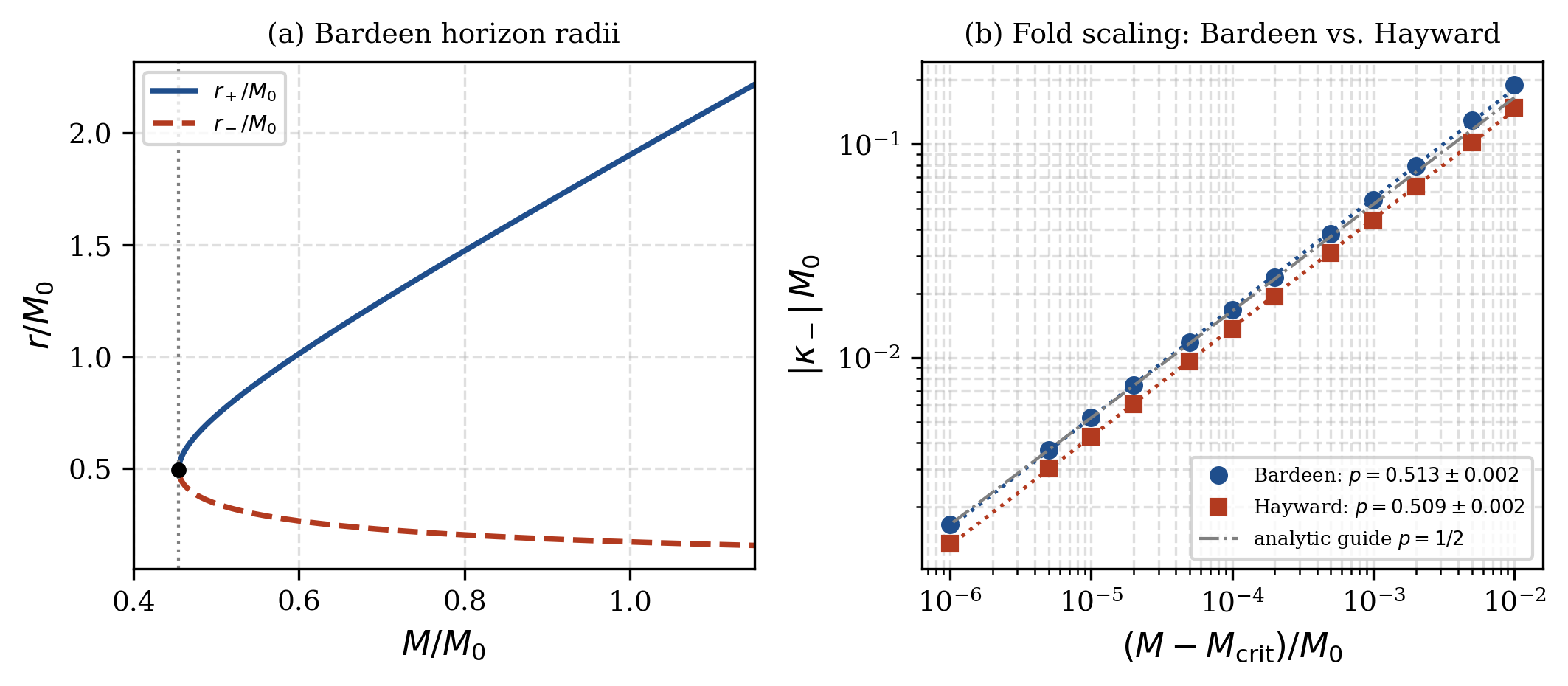}
\caption{(a) Outer and inner horizon radii of the Bardeen metric, Eq.~\eqref{eq:f_bardeen}, for $g=0.35\,M_0$; the two branches merge at $M_{\rm crit}\simeq0.455\,M_0$ (dotted line, black dot). (b) Inner-horizon surface gravity $|\kappa_-|$ versus $(M-M_{\rm crit})/M_0$ on log-log axes, for \emph{both} the Bardeen (circles) and Hayward (squares) metrics with matched core scale $g=l=0.35\,M_0$, each fitted over the \emph{same} twelve points spanning exactly four decades, $(M-M_{\rm crit})/M_0\in[10^{-6},10^{-2}]$. Least-squares fits give $p_{\rm Bardeen}=0.513\pm0.002$ and $p_{\rm Hayward}=0.509\pm0.002$ (regression error only; fit-window/grid systematic envelopes are $[0.503,0.514]$ and $[0.502,0.509]$ respectively, see text), both close to the $p=1/2$ prediction (gray guide) of Proposition~\ref{prop:fold}, proved under explicit hypotheses in Sec.~\ref{sec:universality}.}
\label{fig:horizons}
\end{figure}

\section{RSET Input and Its Near-Extremal Limit}
\label{sec:rset}

The physical input for our model is the RSET analysis of Ref.~\cite{Arrechea2026}, specialized to the Bardeen background. In a static region with $ds^2=-f(r)\,du\,dv$, the vacuum-subtracted $s$-wave RSET components entering the ALS diagnostic are
\begin{equation}
\label{eq:Tuu}
\langle T_{uu}\rangle_{\rm phys} = \frac{1}{192\pi}\,\frac{B(x_1)-B(r)}{f(r)^2},
\end{equation}
where $x_1$ is the radius at which an outgoing null ray enters the trapped region and
\begin{equation}
\label{eq:B_func}
B(r) \equiv \big[f'(r)\big]^2 - 2f(r)f''(r).
\end{equation}
For a trapped region bounded by an inner horizon of surface gravity $\kappa_-<0$ and advanced-time duration $\Delta v$, ALS show that this diagnostic quantity, evaluated at $r_-$ and evolved into the subsequent (evaporated, flat) region, behaves, for $\Delta v\to\infty$ at fixed background, as
\begin{equation}
\label{eq:flux_growth}
\big|\langle T_{uu}\rangle_{\rm phys}\big|_{r_-} \simeq \frac{|f^{(3)}(r_-)|}{384\pi\,|\kappa_-|}\Big[e^{2|\kappa_-|\Delta v}-1\Big].
\end{equation}
This is the analytic mechanism responsible for the exponential amplification seen numerically in Refs.~\cite{BarenboimFrolovKunstatter2024,BoyanovHilditchSemiao2026}, and for seeding the anti-trapped region within the ALS diagnostic. We refer to Eq.~\eqref{eq:flux_growth} throughout as the ``fixed-background RSET diagnostic quantity,'' rather than as a physical emitted flux, to avoid suggesting it is itself a measured or self-consistently propagated observable.

\emph{A naive but incorrect inference.} One might be tempted to conclude from Eq.~\eqref{eq:flux_growth} that this diagnostic quantity vanishes as $\kappa_-\to0$, since $\kappa_-$ multiplies the exponent. This is false at fixed $\Delta v$. Expanding the bracket for $|\kappa_-|\Delta v\to0$,
\begin{equation}
\label{eq:near_extremal_expansion}
\frac{e^{2|\kappa_-|\Delta v}-1}{|\kappa_-|} = 2\Delta v + O\big(|\kappa_-|\Delta v^2\big), \qquad |\kappa_-|\to0,
\end{equation}
so that $|\langle T_{uu}\rangle_{\rm phys}|_{r_-}\to \tfrac{|f^{(3)}(r_-)|}{192\pi}\Delta v$, a \emph{finite, nonzero} limit, not zero -- \emph{provided} $f^{(3)}(r_-)\neq0$. For the Bardeen metric this prefactor is indeed nonzero at the merger point: a direct computation gives
\begin{equation}
\label{eq:f3r0}
f_{\rm B}^{(3)}(r_0) = -\frac{11\sqrt2}{9\,g^3} \neq 0,
\end{equation}
and, as we verify in Sec.~\ref{sec:universality}, the analogous quantity for Hayward is $f_{\rm H}^{(3)}(r_0)=-8\sqrt3/(9l^3)\neq0$ as well. Equation~\eqref{eq:flux_growth} therefore does not, by itself, imply any suppression of this diagnostic quantity as the geometry approaches extremality \emph{at fixed $\Delta v$}. We stress that Eq.~\eqref{eq:near_extremal_expansion} only shows that $\kappa_-\to0$ alone does not force Eq.~\eqref{eq:flux_growth} to vanish: it compares two distinct limits of the same static-background expression (the ALS result is derived for $\Delta v\to\infty$ at fixed background; Eq.~\eqref{eq:near_extremal_expansion} probes fixed $\Delta v$ as $\kappa_-\to0$), and it is \emph{not} a uniform two-parameter asymptotics, nor a statement about the physical flux in an actual, self-consistently backreacted near-extremal evolution, where $\Delta v_n$, $f^{(3)}(r_-)$, and the background would all vary jointly. We do not claim to have shown that such a physical flux is finite or nonzero -- only that the naive inference to the contrary is unjustified. A first-principles derivation of any suppression would require tracking this joint scaling under a self-consistent backreaction solution -- precisely the calculation ALS identify as open. In Sec.~\ref{sec:cascade} we therefore impose a damping law by hand, as an explicit closure assumption, rather than deriving one from Eq.~\eqref{eq:flux_growth}.

We also stress that Eq.~\eqref{eq:flux_growth} is a two-dimensional, $s$-wave quantity. Converting it into a four-dimensional ADM or Bondi mass loss is not simply a matter of multiplying by the area factor of the $4\pi r^2$ dimensional-reduction convention; it also requires an assumption about the quantum state's normalization, the covariant conservation law relating $\langle T_{uu}\rangle$ to $dM/dv$ along the relevant asymptotic foliation, and (for a realistic astrophysical estimate) four-dimensional greybody factors relating the near-horizon flux to the flux actually received at $\mathscr{I}^+$. None of this is computed in Ref.~\cite{Arrechea2026} or attempted here. We take $\Delta M_n$ below to be a phenomenological, closure-defined proxy for radiated energy in natural units, not a first-principles 4D mass-loss rate, and we label it accordingly throughout.

\section{Fold-Point Scaling: Proposition and Two-Metric Check}
\label{sec:universality}

This section states and proves the paper's one result that holds under general local hypotheses rather than for a specific metric, then verifies it numerically in two regular black hole families. We stress at the outset that the proposition below is a \emph{local} statement about static, sufficiently smooth, one-parameter spherically symmetric metric families satisfying explicit nondegeneracy hypotheses; it is not a claim about the dynamics of regular black holes in general, nor about metric families for which these hypotheses fail.

\begin{proposition}
\label{prop:fold}
Let $f(r;M)$ be $C^3$ jointly in $(r,M)$ near a point $(r_0,M_{\rm crit})$, such that for $M>M_{\rm crit}$ in some neighborhood, $f(r;M)=0$ has exactly two roots $r_-(M)<r_0<r_+(M)$ with $r_\pm(M)\to r_0$ as $M\to M_{\rm crit}^+$, and $f(r;M_{\rm crit})=0$ has $r_0$ as its only root near $r_0$. Define $\kappa(r,M)\equiv\tfrac12\partial_rf(r;M)$ and let $M(r)$ denote the horizon curve implicitly defined by $f(r;M(r))\equiv0$ near $r_0$. Suppose
\begin{enumerate}
\item[(i)] $M'(r_0)=0$ and $M''(r_0)>0$ (a nondegenerate/Morse minimum of $M(r)$ at $r_0$), and
\item[(ii)] $\partial_r\kappa(r_0,M_{\rm crit})\neq0$.
\end{enumerate}
Then, along each horizon branch,
\begin{equation}
\label{eq:fold_expansion}
r_\pm(M) - r_0 = \pm\sqrt{\frac{2(M-M_{\rm crit})}{M''(r_0)}} + O(M-M_{\rm crit}),
\end{equation}
\begin{equation}
\label{eq:generic_p}
\kappa_\pm(M) = \pm C\sqrt{M-M_{\rm crit}} + O(M-M_{\rm crit}), \qquad M\to M_{\rm crit}^+,
\end{equation}
with $C\equiv\partial_r\kappa(r_0,M_{\rm crit})\sqrt{2/M''(r_0)}$; in particular $|\kappa_\pm|\propto(M-M_{\rm crit})^{1/2}$ to leading order.
\end{proposition}

\begin{proof}
Since $f(r_0;M_{\rm crit})=0$ and $r_0$ is a double root of $f(\cdot\,;M_{\rm crit})$ (the two branches $r_\pm$ merge there), $\partial_rf(r_0;M_{\rm crit})=0$ as well, i.e.\ $\kappa(r_0,M_{\rm crit})=0$ identically, independent of any further hypothesis on $f$. Condition (i) makes $M_{\rm crit}=M(r_0)$ a nondegenerate minimum of the $C^3$ function $M(r)$, so Taylor's theorem with remainder gives $M(r)=M_{\rm crit}+\tfrac12M''(r_0)(r-r_0)^2+O((r-r_0)^3)$ near $r_0$. Solving for $r-r_0$ (a standard inverse-function/Puiseux expansion for a fold singularity, valid because $M''(r_0)\neq0$) yields the two real branches
\[
r_\pm(M)-r_0=\pm\sqrt{\frac{2(M-M_{\rm crit})}{M''(r_0)}}\Big[1+O(\sqrt{M-M_{\rm crit}})\Big],
\]
which is Eq.~\eqref{eq:fold_expansion}. Separately, Taylor-expanding $\kappa(r,M)$ to first order about $(r_0,M_{\rm crit})$,
\[
\kappa(r,M) = \partial_r\kappa(r_0,M_{\rm crit})\,(r-r_0) + \partial_M\kappa(r_0,M_{\rm crit})\,(M-M_{\rm crit}) + O(2),
\]
where $\partial_M\kappa(r_0,M_{\rm crit})$ is finite by the assumed $C^3$ smoothness of $f$ (so $\kappa=\tfrac12\partial_rf$ is $C^2$, and its first partial derivatives are continuous, hence bounded near $(r_0,M_{\rm crit})$). Evaluating along $r=r_\pm(M)$ and substituting Eq.~\eqref{eq:fold_expansion}, the first term is $O(\sqrt{M-M_{\rm crit}})$ while the second is $O(M-M_{\rm crit})$, strictly subleading as $M\to M_{\rm crit}^+$. Under hypothesis (ii), the leading term does not vanish, giving Eq.~\eqref{eq:generic_p}.
\end{proof}

Hypotheses (i)--(ii) are local nondegeneracy conditions on $f$ at the specific point $(r_0,M_{\rm crit})$; they fail only if $M(r)$ has a higher-order (non-Morse) critical point there, or if the surface gravity happens to be stationary in $r$ exactly at the merger. Neither failure occurs for generic parameter choices in a given metric family, but we make no claim about metric families or parameter ranges beyond the two checked below, and we do not assert a universal statement about regular black hole horizons as a class.

For Bardeen, $M(r)=(r^2+g^2)^{3/2}/(2r^2)$, so $M''(r_0)=\sqrt3/(2g)>0$ and $\partial_r\kappa(r_0,M_{\rm crit})=1/(3g^2)\neq0$: both hypotheses hold, and Eq.~\eqref{eq:f3r0} confirms the associated $f^{(3)}(r_0)\neq0$ used in Sec.~\ref{sec:rset}. For Hayward, Eq.~\eqref{eq:f_hayward} gives $M(r)=r^3/[2(r^2-l^2)]$, with critical point at $r_0=\sqrt3\,l$, $M_{\rm crit}=\tfrac{3\sqrt3}{4}l$ [Eq.~\eqref{eq:Mcrit_H}]; direct computation gives $M''(r_0)=3\sqrt3/(4l)>0$, $\partial_r\kappa(r_0,M_{\rm crit})=1/(3l^2)\neq0$, and $f_{\rm H}^{(3)}(r_0)=-8\sqrt3/(9l^3)\neq0$. Both metrics satisfy Proposition~\ref{prop:fold}'s hypotheses, and both are therefore expected to show $p=1/2$ to leading order.

Figure~\ref{fig:horizons}(b) confirms this numerically: fitting $|\kappa_-(M)|\propto(M-M_{\rm crit})^p$ over the \emph{same} twelve-point grid spanning four decades, $(M-M_{\rm crit})/M_0\in[10^{-6},10^{-2}]$, for both metrics at matched core scale $g=l=0.35\,M_0$, gives regression fits $p_{\rm Bardeen}=0.513\pm0.002$ and $p_{\rm Hayward}=0.509\pm0.002$. Since the regression error alone does not capture systematic effects, we repeated each fit over three fit windows (the full 12-point range and its 8- and 6-point most-asymptotic subsets) and three root-finding grid resolutions (3000, 6000, 12000 points); the fitted $p$ is completely insensitive to grid resolution at fixed window (root-finding is via bracketed Brent iteration to a fixed tolerance, so grid resolution only affects bracket-finding, not root precision) but varies with fit window, giving envelopes $p_{\rm Bardeen}\in[0.503,0.514]$ and $p_{\rm Hayward}\in[0.502,0.509]$; the most-asymptotic windows lie closest to $1/2$ in both metrics, consistent with the expected $O(\sqrt{M-M_{\rm crit}})$ subleading correction of Eq.~\eqref{eq:generic_p}. We report this envelope rather than a single-fit precision claim. This two-metric agreement is evidence, but not a proof for all regular black holes, that the fold scaling of Eq.~\eqref{eq:generic_p} is a structural feature of horizon mergers satisfying Proposition~\ref{prop:fold}'s hypotheses, rather than an accident of the Bardeen metric specifically.

\section{A Phenomenological Cascade Closure}
\label{sec:cascade}

\paragraph{Scope and evidential status.} The recursion below is a conditional phenomenological closure on the static Bardeen family. It is not obtained by integrating the semiclassical Einstein equations, and it does not determine whether repeated trapped/anti-trapped transitions occur in a self-consistent spacetime. The closure supplies, by assumption, both the relation between the fixed-background RSET diagnostic and the mass update, and the cycle duration implicitly encoded in its efficiency parameter. The exponents and tolerance-dependent cycle counts derived below therefore characterize this closure class, not a universal evaporation law; Sec.~\ref{sec:universality} is the only part of this analysis we consider metric-independent.

We define $x_n\equiv (M_n-M_{\rm crit})/M_0\ge0$ and the normalized damping factor
\begin{equation}
\label{eq:Kdef}
K(M) \equiv \frac{|\kappa_-(M)|}{|\kappa_-(M_0)|} \in (0,1], \qquad M\in(M_{\rm crit},M_0],
\end{equation}
well-defined because $|\kappa_-(M)|$ increases monotonically with $M$ on this interval, which we verify numerically for $g/M_0=0.35$ (Fig.~\ref{fig:horizons}b) without a general proof for arbitrary $g$. For a closure exponent $q>0$, we posit
\begin{equation}
\label{eq:recurrence}
x_{n+1} = x_n\big[1-\eta_0\, K(M_n)^q\big], \qquad 0<\eta_0\,K(M)^q<1,
\end{equation}
with $\eta_0\in(0,1)$ a free efficiency parameter; $\eta_0$ and $q$ are postulated, not calibrated against any independent data (there is currently no multi-transition numerical solution to calibrate them against). The actual, closure-defined mass decrement in cycle $n$ is, unambiguously,
\begin{equation}
\label{eq:DeltaMdef}
\Delta M_n \equiv M_n - M_{n+1} = \eta_0\,K(M_n)^q\,(M_n-M_{\rm crit}),
\end{equation}
which we call ``radiated'' only under the additional 2D-to-4D assumptions of Sec.~\ref{sec:rset}. We introduce no residual or terminal energy bookkeeping; stopping the iteration at $x_n<\varepsilon$ is a numerical convention, not a physical endpoint.

\paragraph{General $(p,q)$ classification.} Using the fold scaling $K(M)\simeq\beta x^p$ from Sec.~\ref{sec:universality} (with $p=1/2$ for both metrics studied, $x\equiv(M-M_{\rm crit})/M_0$) in the closure~\eqref{eq:recurrence}, the continuum limit is
\begin{equation}
\label{eq:continuum_general}
\frac{dx}{dn} \simeq -\eta_0\beta^q\, x^{1+pq},
\end{equation}
whose solution gives the asymptotic power law
\begin{equation}
\label{eq:powerlaw_general}
x_n \;\sim\; \Big(\tfrac{pq}{2}\,\eta_0\beta^q\, n\Big)^{-1/(pq)}, \qquad n\to\infty,
\end{equation}
so that the tolerance-crossing cycle count scales as $N(\varepsilon)\sim\varepsilon^{-pq}$. Our numerical results use the linear closure $q=1$, giving the special case
\begin{equation}
\label{eq:powerlaw}
x_n \sim \frac{4}{(\eta_0\beta\, n)^2}, \qquad N(\varepsilon)\sim\varepsilon^{-1/2}, \qquad (p=\tfrac12,\,q=1).
\end{equation}
Equation~\eqref{eq:powerlaw_general} is the paper's main classification statement: the specific exponent $-2$ in Eq.~\eqref{eq:powerlaw} is not an independent physical input but a consequence of combining the (here, doubly-verified) fold exponent $p=1/2$ with one particular, unverified choice of closure exponent $q=1$. Under \emph{any} closure in this class, the sequence approaches $M_{\rm crit}$ only asymptotically and never reaches it after a model-independent finite number of cycles; any quoted $N$ is a property of $(\eta_0,q,\varepsilon)$, not a physical prediction. We report numerical results only for $q=1$ below.

\paragraph{Step-size versus adiabaticity.} Treating $f(r;M_n)$ as static within each cycle is a modeling choice, and the cited mechanism is an \emph{exponential} amplification that could, in principle, produce large changes per cycle. Because no cycle in Eq.~\eqref{eq:recurrence} carries an assigned physical duration $\Delta v_n$ -- it is absorbed into $\eta_0$ -- we cannot form a genuine rate-based adiabaticity criterion (e.g.\ comparing $|\dot M|/(M|\kappa_+|)$ to unity) from the discrete recursion alone. What we can and do check, in Sec.~\ref{sec:numerics}, is the more limited fractional per-cycle update $\Delta M_n/M_n$, which is of order unity for the first few cycles and falls rapidly thereafter (Fig.~\ref{fig:adiabaticity}). We report this as a step-size diagnostic of the recursion itself, not as a physical adiabaticity test.

\section{Numerical Protocol and Results}
\label{sec:numerics}

\paragraph{Numerical protocol.} At every step we solve $f(r;M_n)=0$ only for $M_n>M_{\rm crit}$. Roots are located in two stages: (i) $f(r;M_n)$ is evaluated on a uniform grid of $4000$--$6000$ points on each of the finite intervals $[10^{-6},r_0]$ and $[r_0,r_{\max}]$, with $r_{\max}=6\,M_0$ (chosen so that $r_+(M_0)\ll r_{\max}$ with a comfortable margin; doubling $r_{\max}$ leaves all reported results unchanged to the quoted precision); (ii) within the bracketing grid interval where $f$ changes sign, Brent's method (\texttt{scipy.optimize.brentq}, absolute tolerance $10^{-13}$ in $r$) refines the root. We accept a root only if the residual satisfies $|f(r_\pm;M_n)|<10^{-10}$; no root failed this check in any run reported here. In pseudocode,
\begin{verbatim}
def horizons(M):
  grid_lo = linspace(1e-6, r0, N_grid)
  grid_hi = linspace(r0, 6*M0, N_grid)
  find sign change in f(grid_lo,M) -> bracket
  r_minus = brentq(f, bracket, xtol=1e-13)
  (repeat on grid_hi for r_plus)
  assert |f(r_minus,M)|<1e-10
  assert |f(r_plus,M)|<1e-10
  return r_minus, r_plus
\end{verbatim}
The cascade recursion, Eq.~\eqref{eq:recurrence}, is iterated forward from $M_0$ with $q=1$, calling \texttt{horizons} once per cycle to evaluate $\kappa_-(M_n)$. We use $g=0.35\,M_0$ (so $M_{\rm crit}\simeq0.4547\,M_0$) throughout, and report results for several $(\eta_0,\varepsilon)$ pairs explicitly rather than a single preferred choice. For the timescale comparison of Sec.~\ref{sec:vaidya} we additionally need $\kappa_\pm(M)$ reliably down to $x=(M-M_{\rm crit})/M_0\sim10^{-7}$, where the true horizon splitting $|r_\pm-r_0|\propto\sqrt{x}$ becomes comparable to, and then smaller than, the fixed grid spacing $r_0/N_{\rm grid}$ used above; there, we instead bracket each root within an adaptive window of half-width $\sim5\sqrt{2x/M''(r_0)}$ about $r_0$ (doubled until a sign change is found), which we have checked reproduces the fixed-grid results wherever both are reliable ($x\gtrsim10^{-3}$) and remains well-behaved down to $x\sim10^{-7}$. The complete implementation (Python, \texttt{numpy}/\texttt{scipy}/\texttt{matplotlib}) used to produce every figure in this paper, including the Hayward comparison of Sec.~\ref{sec:universality}, the fit-window/grid systematic-uncertainty scan of Fig.~\ref{fig:horizons}(b), and the adaptive-bracket calculation of Sec.~\ref{sec:vaidya}, is provided as supplementary material alongside the manuscript; the adaptive bracket is used by default for all near-extremal ($x\lesssim10^{-3}$) quantities reported in Secs.~\ref{sec:universality} and \ref{sec:vaidya}, with the fixed grid retained only for the coarser cascade iteration of Sec.~\ref{sec:cascade} where $x$ does not approach this regime as closely.

\begin{figure*}[t]
\centering
\includegraphics[width=0.98\textwidth]{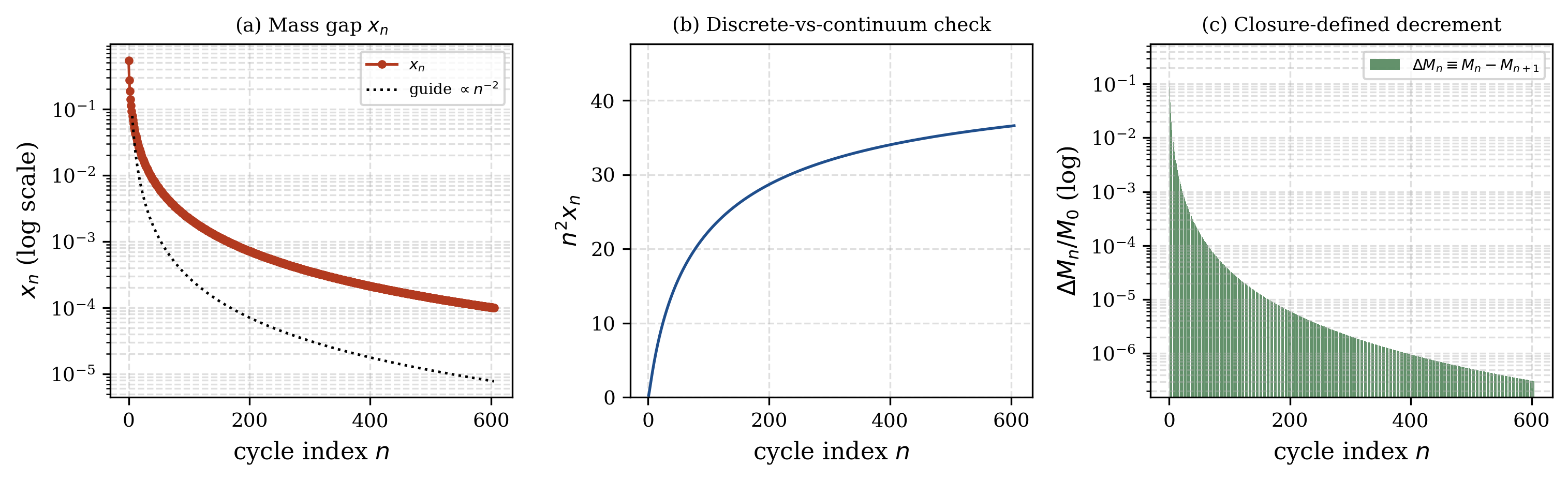}
\caption{(a) Mass gap $x_n=(M_n-M_{\rm crit})/M_0$ versus cycle index $n$ on a logarithmic scale, for $\eta_0=0.5$, $q=1$, $\varepsilon=10^{-4}$ ($N=605$ cycles); the dotted guide $\propto n^{-2}$ illustrates the asymptotic power law of Eq.~\eqref{eq:powerlaw}. (b) Discrete-vs-continuum check: $n^2x_n$ computed directly from the discrete recursion (no fit), which should approach a constant if Eq.~\eqref{eq:powerlaw} correctly captures the leading late-cycle behavior. The quantity is still rising slowly at $n=605$ ($n^2x_n\to$ a few tens over $n\sim10^2$--$10^3$; a longer supplementary run at looser tolerance shows it continues to rise slowly out to $n\sim6000$ before flattening near $\sim45$), consistent with a leading $n^{-2}$ term plus slowly-decaying subleading corrections; we have not fitted or derived the precise functional form of this subleading correction. (c) Closure-defined mass decrement per cycle, $\Delta M_n=M_n-M_{n+1}$ [Eq.~\eqref{eq:DeltaMdef}], on a logarithmic axis, with no residual or terminal energy added.}
\label{fig:mass_decay}
\end{figure*}

Figure~\ref{fig:mass_decay}(a) shows $x_n$ for $\eta_0=0.5$, $\varepsilon=10^{-4}$: reaching this tolerance requires $N=605$ cycles. Panel (b) provides a direct, fit-free check of the asymptotic scaling: rather than only overlaying a guide line, we plot $n^2x_n$ computed from the actual discrete sequence; it approaches a constant only slowly (still rising at $n=605$, and, as we checked in a longer supplementary run, continuing to rise slowly out to $n\sim6000$ before flattening near $\sim45$), which is what we expect from a leading $n^{-2}$ tail modulated by slowly-decaying subleading corrections intrinsic to the discrete-to-continuum map, not evidence against Eq.~\eqref{eq:powerlaw}. Panel (c) shows the closure-defined decrement $\Delta M_n$: the first cycle removes $\Delta M_0\simeq0.27\,M_0$, and subsequent cycles remove a rapidly decreasing amount, spanning more than five orders of magnitude by $n\sim600$.

\begin{figure}[t]
\centering
\includegraphics[width=0.49\textwidth]{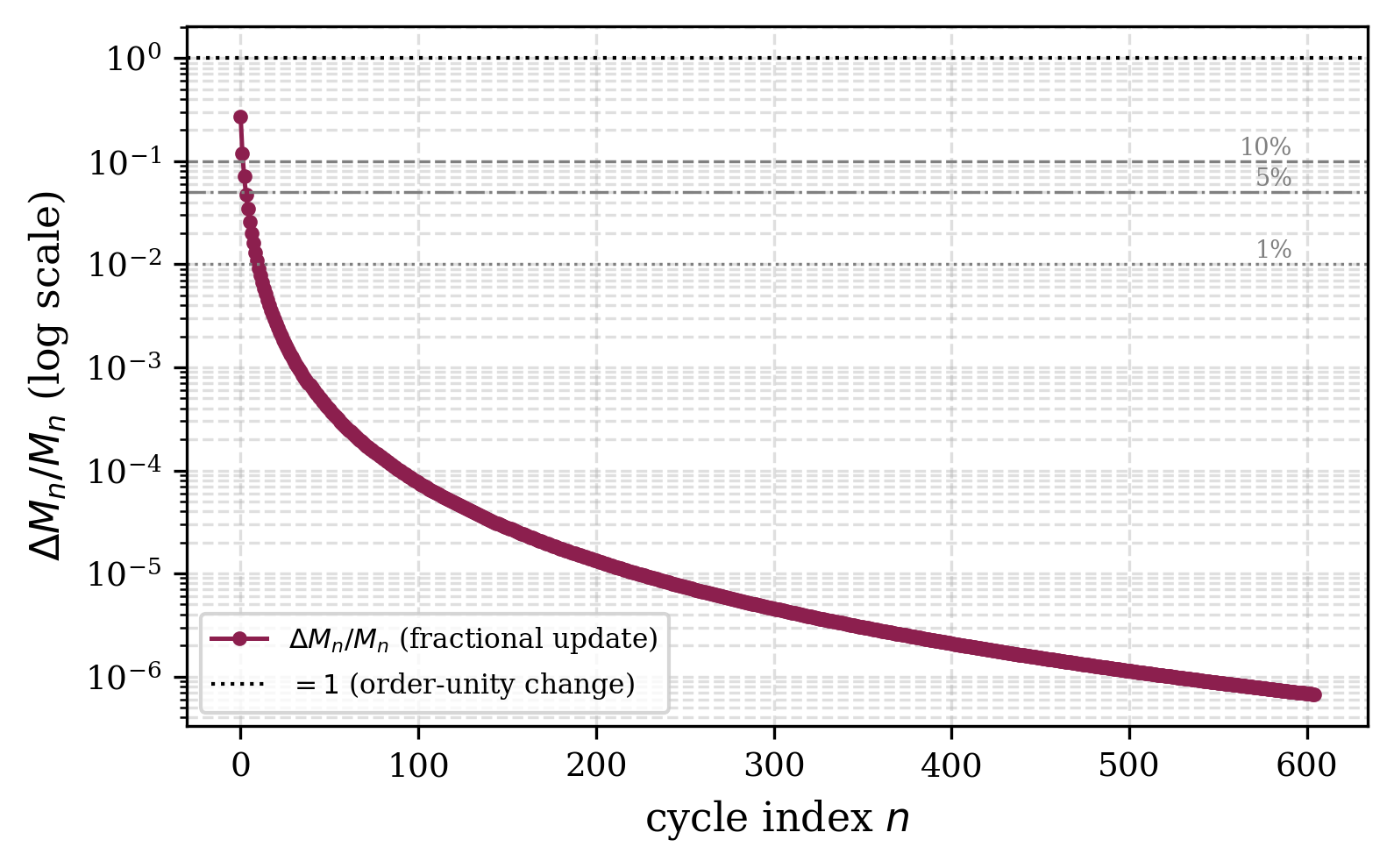}
\caption{Fractional per-cycle mass update, $\Delta M_n/M_n$, on a logarithmic scale, for the same run as Fig.~\ref{fig:mass_decay}, with several reference thresholds shown. This is a step-size diagnostic for the discrete recursion, not a physical adiabaticity test: no cycle in Eq.~\eqref{eq:recurrence} carries an assigned physical duration $\Delta v_n$ (it is absorbed into $\eta_0$), so a rate-based criterion cannot be formed from the cycle index alone. Of 605 cycles, 2 exceed 10\%, 3 exceed 5\%, and 10 exceed 1\% (the first cycle alone updates $M$ by 27\%); beyond $n\sim10$, $\Delta M_n/M_n$ falls several orders of magnitude below unity.}
\label{fig:adiabaticity}
\end{figure}

Figure~\ref{fig:adiabaticity} makes the step-size behavior of the closure quantitative: only a handful of the 605 cycles update $M$ by more than a few percent, with the first cycle alone responsible for a 27\% change -- an order-unity update by any threshold. We do not interpret this as a physical adiabaticity test (see above), but simply as showing that the recursion's own step size shrinks quickly. In this limited step-size sense, the closure's early cycles are its least smoothly-varying ones; the classification result of Eq.~\eqref{eq:powerlaw_general} describes the tail where the steps are already small, while the largest single contributions to the total mass loss come from the earliest, largest-step cycles.

\begin{figure}[t]
\centering
\includegraphics[width=0.8\textwidth]{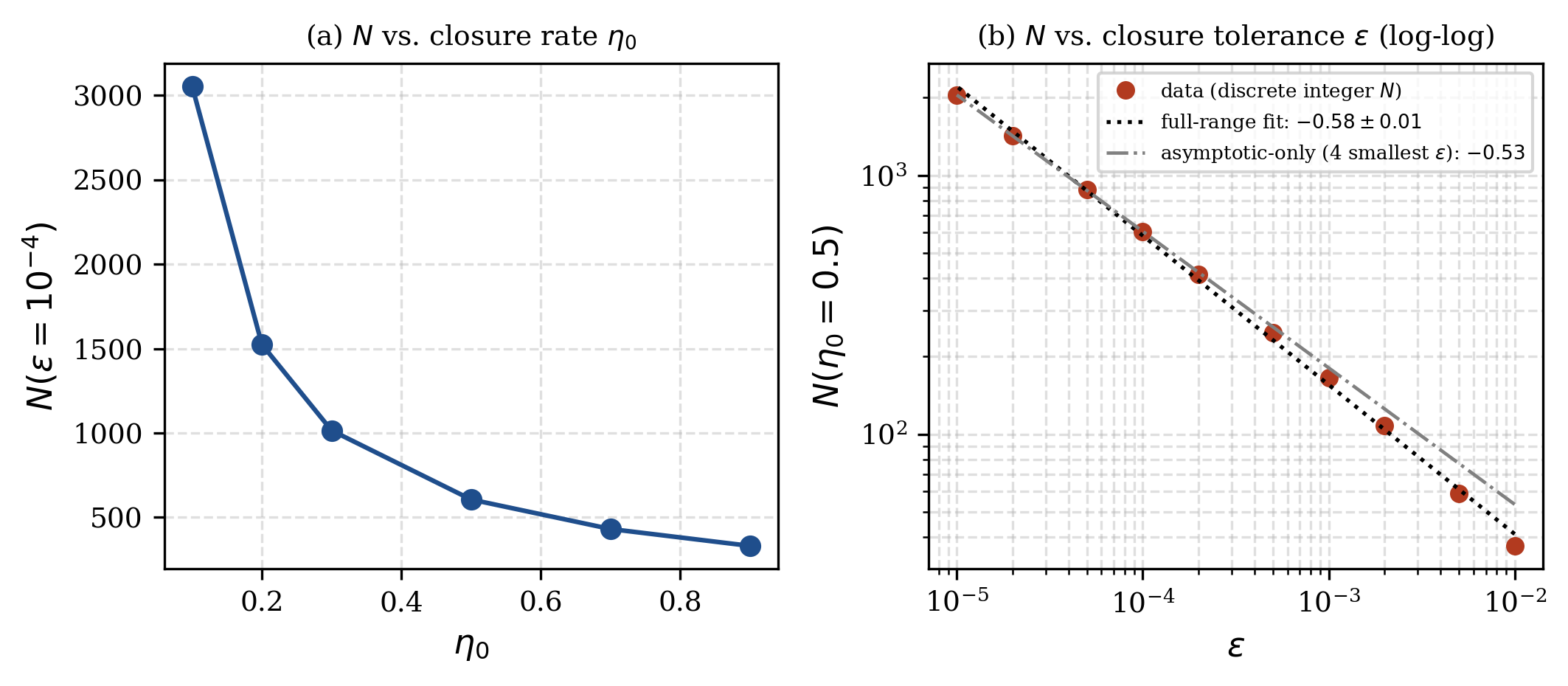}
\caption{Sensitivity of the cycle count $N$ to the model's free choices, over $\eta_0\in\{0.1,\dots,0.9\}$ (6 samples) and $\varepsilon\in\{10^{-2},\dots,10^{-5}\}$ (10 log-spaced samples), at $q=1$. (a) $N$ at fixed $\varepsilon=10^{-4}$ vs.\ $\eta_0$: an order-of-magnitude variation. (b) $N$ at fixed $\eta_0=0.5$ vs.\ $\varepsilon$ (integer-valued data, log-log). An unweighted full-range fit gives slope $-0.58\pm0.01$; restricting to the four smallest (most asymptotic) $\varepsilon$ gives $-0.53$, markedly closer to the analytic prediction $-1/2$ of Eq.~\eqref{eq:powerlaw}. We read the full-range/asymptotic-only difference as a real pre-asymptotic correction, not fit noise, and avoid claiming the fit ``confirms'' $-1/2$ outright.}
\label{fig:sensitivity}
\end{figure}

Figure~\ref{fig:sensitivity} quantifies the closure's non-universality directly: at fixed $\varepsilon$, $N$ varies by close to an order of magnitude as $\eta_0$ ranges over the sampled values; at fixed $\eta_0$, the asymptotic-only fit ($-0.53$) sits markedly closer to $-1/2$ than the full-range fit ($-0.58\pm0.01$), exactly the pre-asymptotic-correction pattern anticipated from the slow convergence seen directly in Fig.~\ref{fig:mass_decay}(b). We regard this scan, together with the discrete-vs-continuum check, as the honest numerical content of the model: internally consistent, but not indicative of a universal cycle count.

\section{An Effective Adiabatic Model for Outer-Horizon Mass Loss}
\label{sec:vaidya}

The closure of Sec.~\ref{sec:cascade} postulates both the functional form of the mass loss per cycle and its two free parameters $(\eta_0,q)$. Here we construct, for the narrower question of how the \emph{outer} horizon mass evolves within a single trapped phase, an effective adiabatic model tied to the cited ALS flux rather than to a postulated efficiency. We state the status of this model precisely, since a naive area-factor argument alone is not, by itself, sufficient to justify it.

\paragraph{Status of the model.} Matching a two-dimensional Polyakov stress tensor to a four-dimensional source in a time-dependent metric is not simply a matter of undoing the $s$-wave reduction; it requires choosing a specific metric ansatz for the time dependence, computing its exact Einstein tensor, and specifying which part of that tensor is attributed to the additional semiclassical flux as opposed to whatever classical matter already supports the static geometry at each instant. Below we make this explicit for the simplest available ansatz -- the core scale $g$ held fixed while $M\to M(v)$, i.e.\ a generalized Vaidya-Bardeen metric -- and we compute its Einstein tensor exactly rather than assuming the vacuum Schwarzschild-Vaidya relation. The result is an effective adiabatic prescription, not a first-principles non-adiabatic solution of the semiclassical Einstein equations; slowly evaporating black holes are known to be well approximated by quasi-static sequences of static metrics quite generally~\cite{DahalSimovic2023}, which is the sense in which we use it here.

\paragraph{Exact Einstein tensor of the generalized Vaidya-Bardeen ansatz.} Consider $ds^2=-f(r,v)dv^2+2dv\,dr+r^2d\Omega^2$ with
\begin{equation}
\label{eq:vaidya_bardeen}
f(r,v) = 1 - \frac{2M(v)r^2}{(r^2+g^2)^{3/2}},
\end{equation}
$g$ fixed and $M(v)$ an arbitrary function of advanced time. A direct symbolic computation of the Einstein tensor (Christoffel symbols, Riemann and Ricci tensors, and $G_{vv}=R_{vv}-\tfrac12g_{vv}R$, all evaluated in closed form) gives, collecting terms in $\dot M\equiv dM/dv$,
\begin{equation}
\label{eq:Gvv_exact}
G_{vv} = \frac{2r}{(r^2+g^2)^{3/2}}\,\dot M \;+\; G_{vv}^{(0)}(r,M),
\end{equation}
where $G_{vv}^{(0)}$ is the $\dot M$-independent remainder, a function of $r$ and the instantaneous $M$ only. At $g=0$ this reduces to the standard Schwarzschild-Vaidya result $G_{vv}=2\dot M/r^2$; at $g\neq0$ the coefficient of $\dot M$ differs from $2/r^2$ by the factor
\begin{equation}
\label{eq:Zfactor}
Z(r) \equiv \frac{(r^2+g^2)^{3/2}}{r^3}, \qquad \frac{2r}{(r^2+g^2)^{3/2}} = \frac{2}{r^2}\,\frac{1}{Z(r)},
\end{equation}
which is an $O(1)$ correction wherever $g/r$ is not small: at $r=r_+(M_0)$ for our fiducial parameters, $Z\simeq1.05$, rising to $Z(r_0)=(3/2)^{3/2}\simeq1.84$ at the fold. We identify $G_{vv}^{(0)}$ with the Einstein tensor already sourced, at fixed $M$, by the classical nonlinear-electrodynamics stress tensor of the static Bardeen solution~\cite{Ayon-BeatoGarcia2000}: it is present whether or not $M$ varies, and requires no additional (quantum) flux to support. What requires an additional source is specifically the $\dot M$-proportional piece, which we match to the semiclassical flux via $\big[2r/(r^2+g^2)^{3/2}\big]\dot M = 8\pi T_{vv}^{\rm (4D),\,RSET}=(2/r^2)\langle T_{vv}\rangle^{\rm (2D)}$, using the ALS reduction convention for the last equality. Solving for $\dot M$,
\begin{equation}
\label{eq:Tvv_ALS}
\frac{dM}{dv} = \langle T_{vv}\rangle_{\rm phys}\big|_{r_+(M)} \, Z\big(r_+(M)\big) = -\frac{\kappa_+(M)^2}{48\pi}\,Z\big(r_+(M)\big),
\end{equation}
using $\langle T_{vv}\rangle_{\rm phys}|_{r_h}=-\kappa_h^2/(48\pi)$ from ALS. Equation~\eqref{eq:Tvv_ALS} is the corrected effective mass-loss law; it reduces to the naive (uncorrected) relation only in the formal limit $g\to0$. We stress that isolating the $\dot M$-linear piece of $G_{vv}$ as ``the semiclassical flux'' is itself part of the adiabatic ansatz -- a genuinely non-adiabatic evolution could redistribute stress-energy between the core and the additional flux in ways this decomposition does not capture -- and that Eq.~\eqref{eq:Tvv_ALS} says nothing about whether the conserved combination $\nabla_aT^{ab}=0$ for the total effective source is separately satisfied by the core and additional pieces individually; we have not verified this beyond the level of the $G_{vv}$ component used here.

\begin{figure}[t]
\centering
\includegraphics[width=0.49\textwidth]{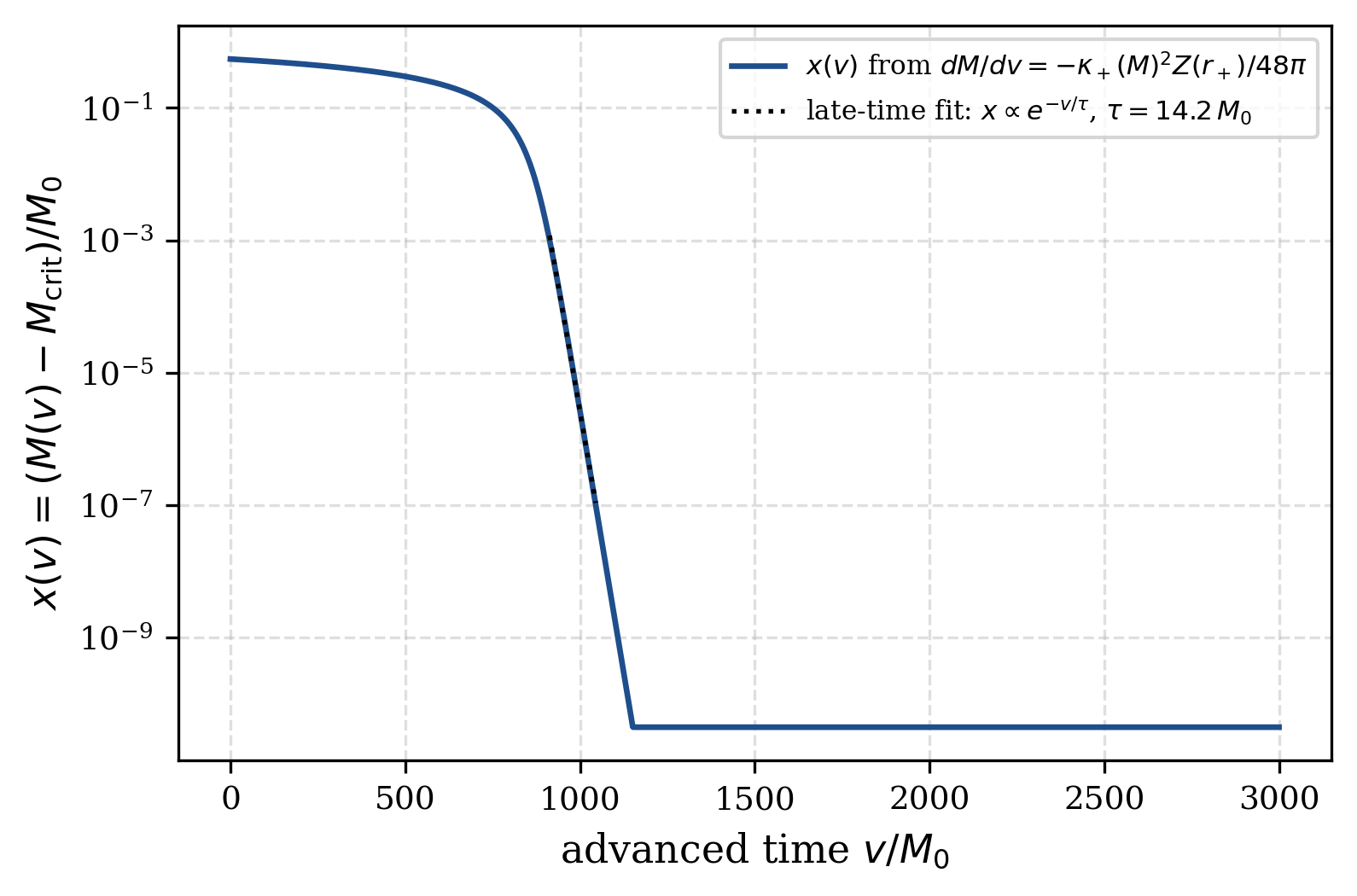}
\caption{Numerical solution of the effective adiabatic law~\eqref{eq:Tvv_ALS}, $x(v)=(M(v)-M_{\rm crit})/M_0$ versus advanced time $v/M_0$, for $g=0.35\,M_0$, $M(0)=M_0$. The approach to $M_{\rm crit}$ is slow at first ($\kappa_+$ is small for $M$ near $M_0$; cf.\ Fig.~\ref{fig:horizons}b), accelerates once $\kappa_+(M)$ grows, and becomes a clean exponential relaxation, $x\propto e^{-v/\tau}$ with $\tau=14.2\,M_0$, once $\kappa_+^2\propto x$ near the fold (Sec.~\ref{sec:universality}); $Z(r_+)$ approaches the finite constant $(3/2)^{3/2}$ there and does not change this qualitative behavior. The curve reaches $x<10^{-4}$ by $v\simeq949\,M_0$ and falls below our numerical floor of $x\sim2\times10^{-10}$ (set by the terminal condition in the integrator, not a physical endpoint) by $v\simeq1300\,M_0$; Eq.~\eqref{eq:Tvv_ALS} itself only guarantees an asymptotic approach, never an exact finite-$v$ endpoint.}
\label{fig:vaidya}
\end{figure}

Figure~\ref{fig:vaidya} shows the numerical solution of Eq.~\eqref{eq:Tvv_ALS}. As in the $g\to0$ (Schwarzschild) limit of this calculation, the late-time behavior is a clean exponential relaxation rather than a power law: since $\kappa_+(M)\propto\sqrt{M-M_{\rm crit}}$ near threshold (Sec.~\ref{sec:universality}) while $Z(r_+)\to(3/2)^{3/2}$ stays finite and nonzero, Eq.~\eqref{eq:Tvv_ALS} still linearizes to $dx/dv\simeq-{\rm const}\times x$ near $M_{\rm crit}$, giving $x(v)\sim e^{-v/\tau}$ with $\tau=14.2\,M_0$ fitted directly from the $Z$-corrected solved trajectory -- about a factor of $1.8$ shorter than the value obtained from the uncorrected, $g\to0$ relation, illustrating that the core-curvature correction is quantitatively, though not qualitatively, significant. This is different from the power-law-in-cycle-number tail $x_n\sim n^{-2}$ of the discrete closure in Sec.~\ref{sec:cascade}, since $v$ and $n$ are different variables governed by different (and here, deliberately not reconciled) dynamics.

More importantly: this effective adiabatic model does not reproduce the finite-advanced-time disappearance of the trapped region reported in fully non-adiabatic numerical solutions~\cite{BarenboimFrolovKunstatter2024,BoyanovHilditchSemiao2026}. Equation~\eqref{eq:Tvv_ALS} only ever approaches $M_{\rm crit}$ asymptotically, even though it does so rapidly in practice ($v\lesssim1200\,M_0$ to reach $x<10^{-5}$). We do \emph{not} read this discrepancy as establishing that the finite-time closure found in genuine PDE solutions is specifically a non-adiabatic effect: the mismatch could equally reflect the choice of quantum state, the dimensional-reduction convention, the specific Vaidya-Bardeen ansatz, the core matter's response to time-dependence, boundary conditions, or details of the cited numerical setups, none of which this adiabatic model controls for. What the comparison does show, robustly, is that this particular effective model -- built from the cited fixed-background flux under the stated adiabatic ansatz -- is insufficient by itself to reproduce the finite-time result, which is consistent with (but does not prove) the point already made in Sec.~\ref{sec:cascade} that adiabatic reasoning about the earliest, largest-change part of this problem is not obviously self-consistent (Fig.~\ref{fig:adiabaticity}).

We emphasize what this section does and does not accomplish. It replaces the postulated $(\eta_0,q)$ closure with an explicit, effective mass-update law for one specific quantity (the outer-horizon mass, under a stated adiabatic ansatz and a specific generalized-Vaidya ansatz for the time dependence), derived from the exact Einstein tensor of that ansatz matched to the cited RSET flux. It does not solve the semiclassical Einstein equations for the full trapped-to-anti-trapped-to-trapped cascade, does not verify the conservation of the total effective source beyond the $G_{vv}$ component used, does not determine $\Delta v_n$ for the anti-trapped phase, and does not track the outgoing-flux-driven formation of the anti-trapped region itself. A full multi-cycle self-consistent calculation, along the lines of the single-transition results of Ref.~\cite{BoyanovHilditchSemiao2026} extended through a second transition, remains the open problem identified throughout this paper.

\paragraph{A timescale comparison.} With this caveat in mind, we can still ask, using only the quantities defined above, which of two mechanisms -- outer-horizon evaporation under the adiabatic model, or inner-horizon amplification under the ALS fixed-background diagnostic -- operates faster at a given mass, sharpening the ``faster than Hawking'' statements made qualitatively in the literature~\cite{BarenboimFrolovKunstatter2024,BoyanovHilditchSemiao2026,BarceloCarballoRubio2016}. Define $\tau_-(M)\equiv1/(2|\kappa_-(M)|)$, the $e$-folding time of the ALS growth factor in Eq.~\eqref{eq:flux_growth}, and $\tau_{\rm evap}(M)\equiv(M-M_{\rm crit})/|dM/dv|$ from Eq.~\eqref{eq:Tvv_ALS}. Both are fully determined by the static Bardeen formulas in hand, with no additional free parameter beyond the effective-model status of $\tau_{\rm evap}$ already flagged.

\begin{figure}[t]
\centering
\includegraphics[width=0.8\textwidth]{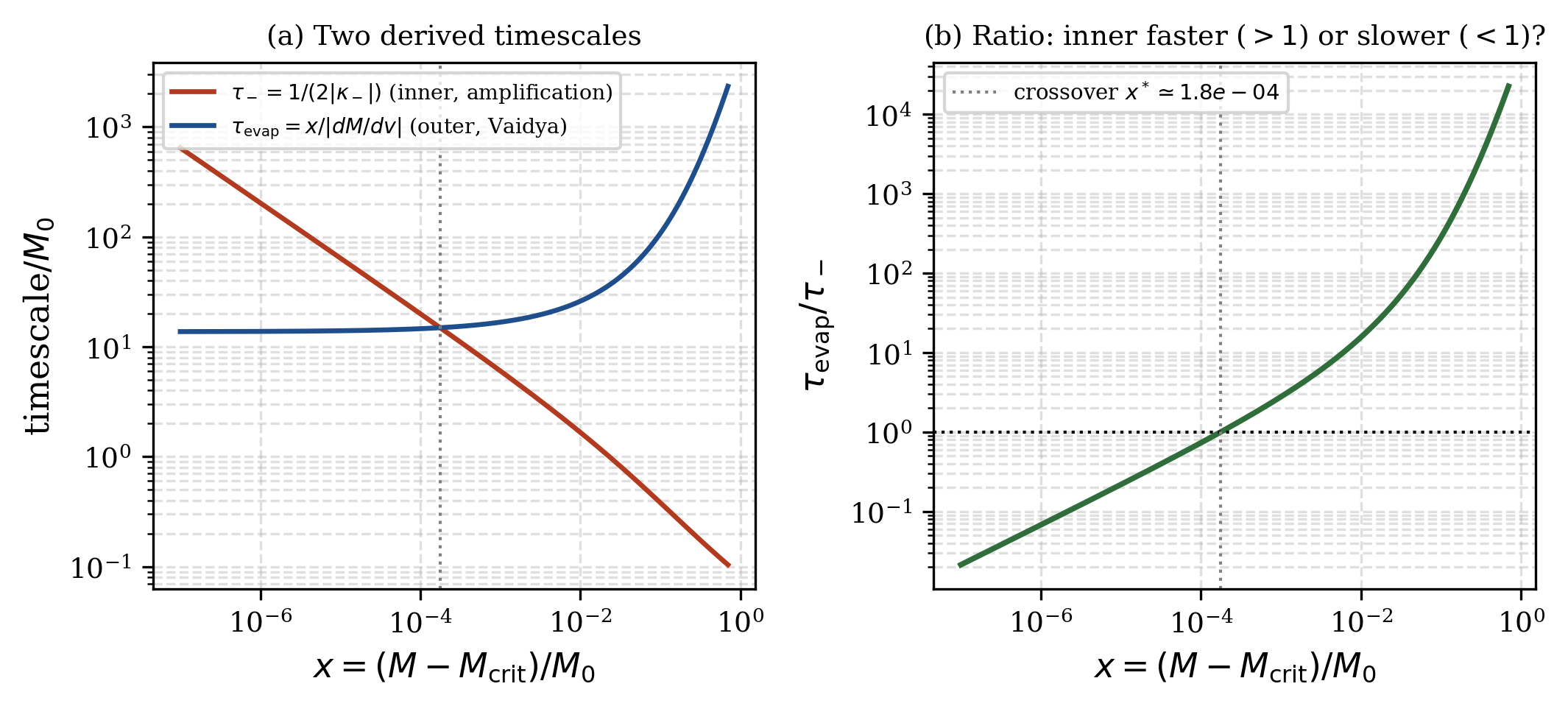}
\caption{(a) The two timescales $\tau_-(M)$ and $\tau_{\rm evap}(M)$ versus $x=(M-M_{\rm crit})/M_0$, log-log; $\tau_{\rm evap}$ uses the $Z$-corrected law, Eq.~\eqref{eq:Tvv_ALS}. (b) Their ratio. Away from extremality, $\tau_-\ll\tau_{\rm evap}$ by up to four orders of magnitude at $x=x_0\simeq0.545$ (ratio $\simeq1.1\times10^4$): under the adiabatic model, the inner-horizon amplification mechanism is parametrically faster than outer-horizon evaporation, consistent with the numerically reported ``inside-out'' evaporation faster than $O(M^3)$. Extremely close to extremality, $x\lesssim x^*\simeq1.8\times10^{-4}$ (for $g=0.35\,M_0$), the ordering reverses: $\kappa_-\to0$ quenches the amplification mechanism itself faster (in $x$) than the adiabatic model's evaporation rate is quenched, so $\tau_-\gg\tau_{\rm evap}$ there. Resolving this crossover numerically requires bracketing $r_\pm$ with a window adapted to the expected $\sqrt{x}$ horizon splitting (Sec.~\ref{sec:numerics}).}
\label{fig:timescales}
\end{figure}

Figure~\ref{fig:timescales} shows both timescales and their ratio. For most of the range explored (order-unity $x$ down to $x\sim10^{-4}$), $\tau_-\ll\tau_{\rm evap}$: the inner-horizon mechanism operates faster than the outer-horizon one under this model, by up to four orders of magnitude. We read this as an illustration, internal to the adiabatic model, of why the trapped region could plausibly disappear ``inside-out'' well before an $O(M^3)$ timescale, consistent with the numerical results of Refs.~\cite{BarenboimFrolovKunstatter2024,BoyanovHilditchSemiao2026}; it is not a dynamical proof that the anti-trapped region forms, which remains tied to the outgoing-flux diagnostic and the caveats above. The ordering reverses very close to extremality: because $\tau_{\rm evap}$ approaches the finite constant $\tau=14.2\,M_0$ near threshold while $\tau_-\propto x^{-1/2}$ continues to grow, there is a crossover at $x^*\simeq1.8\times10^{-4}$ below which $\tau_-\gg\tau_{\rm evap}$. This crossover is a feature of the specific adiabatic model and metric family used here, not an independently verified physical prediction; we flag it as a quantitative feature that a genuine non-adiabatic treatment of the last stage of the cascade should address, confirm, or revise.

\section{An Illustrative Radiation Profile}
\label{sec:phenomenology}

\begin{figure}[t]
\centering
\includegraphics[width=0.6\textwidth]{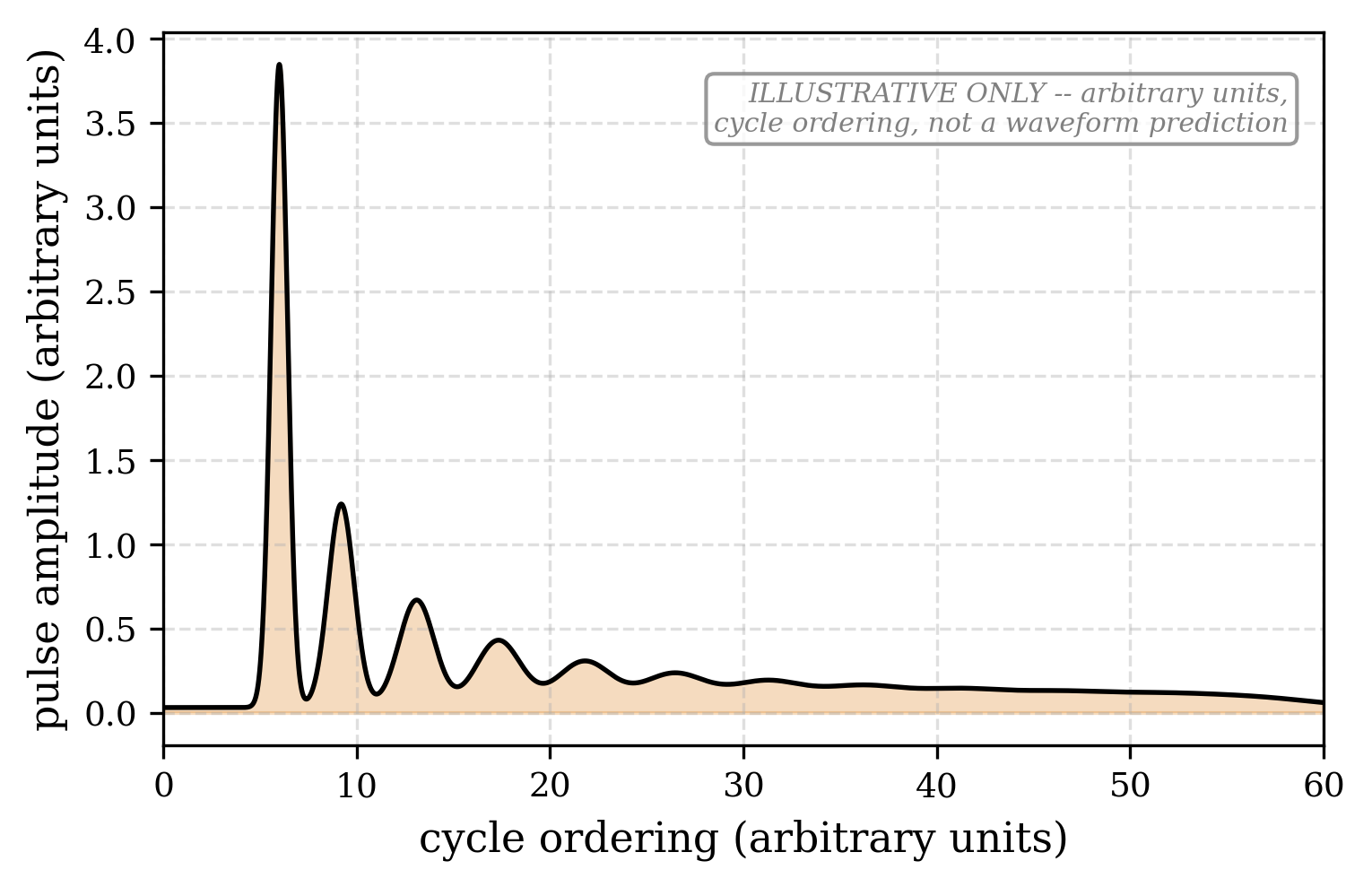}
\caption{Purely schematic representation of the closure-defined decrements $\Delta M_n$ [Eq.~\eqref{eq:DeltaMdef}] as a sequence of pulses, with amplitude $\propto\Delta M_n$ and arbitrarily increasing width, for the first 12 (best-resolved) cycles. Axes are deliberately labeled in arbitrary units and cycle ordering, \emph{not} physical retarded time, since neither the pulse shape, centers, nor widths are derived from a transport equation; this panel carries no quantitative predictive content.}
\label{fig:bursts}
\end{figure}

If each cycle's closure-defined decrement $\Delta M_n$ is eventually released as a pulse toward $\mathscr{I}^+$, one might picture a train of pulses of decreasing amplitude. Figure~\ref{fig:bursts} sketches this using $\Delta M_n$ from Fig.~\ref{fig:mass_decay}(c), for the first twelve cycles. We emphasize once more that this is a cartoon, not a propagated solution: an actual burst waveform would require solving the transport of Eq.~\eqref{eq:Tuu} to $\mathscr{I}^+$ through the intervening (backreacted) geometry, following the methods of Ref.~\cite{Arrechea2026}, which we do not attempt here; we have therefore relabeled the axes to avoid any impression of a physical waveform prediction.

\section{Conclusions}
\label{sec:conclusions}

We have analyzed a minimal, explicitly phenomenological closure for a hypothetical cascade of semiclassical black-hole--white-hole transitions, built on the single-transition RSET diagnostic of Arrechea, Liberati and Spadafora~\cite{Arrechea2026}. Two results are, we believe, robust within the static diagnostic used here, independent of whether the closure of Sec.~\ref{sec:cascade} describes real backreaction: (i) the ALS diagnostic quantity, Eq.~\eqref{eq:flux_growth}, does \emph{not} by itself imply any damping near extremality at fixed cycle duration, since $f^{(3)}(r_0)\neq0$ in both metrics checked, Eq.~\eqref{eq:f3r0}; and (ii) the inner-horizon surface gravity vanishes as a square root of the mass gap, $\kappa_-\propto(M-M_{\rm crit})^{1/2}$, which Sec.~\ref{sec:universality} shows follows from explicit nondegeneracy hypotheses on any fold-type horizon merger (Proposition~\ref{prop:fold}), not from any Bardeen-specific detail, and which we confirm in two independent metrics with a quantified fit-window/grid systematic envelope in addition to the regression error. Neither statement has been shown to survive a self-consistent, dynamically backreacted evolution; both are established here only for static backgrounds under the proposition's stated hypotheses.

Building one linear ($q=1$) closure on the verified $p=1/2$ scaling gives a cascade that approaches $M_{\rm crit}$ only as a power law, $x_n\sim n^{-2}$, which we have generalized to the classification $x_n\sim n^{-1/(pq)}$, $N(\varepsilon)\sim\varepsilon^{-pq}$, for arbitrary fold and closure exponents [Eq.~\eqref{eq:powerlaw_general}]; the closure exponent $q$ and efficiency $\eta_0$ are postulated, not calibrated, so this classification describes the properties of a chosen recursion class, not a physical prediction. We verified the $q=1$ special case against the raw discrete recursion (not only a fitted guide) and tracked a fractional per-cycle update diagnostic, which is not itself a physical adiabaticity test (Figs.~\ref{fig:mass_decay}--\ref{fig:adiabaticity}). As a separate exercise, Sec.~\ref{sec:vaidya} constructed an effective adiabatic model for the outer-horizon mass from the exact Einstein tensor of a generalized Vaidya-Bardeen ansatz matched to the cited RSET flux -- a model, not a first-principles derivation, since it rests on a specific ansatz for the time dependence and an assumption about which part of the source is ``additional'' -- finding an exponential-in-$v$ relaxation to $M_{\rm crit}$ that does not by itself reproduce the finite-time transition seen in non-adiabatic simulations. Comparing the resulting timescale with the inner-horizon amplification rate set by $\kappa_-$, the latter is faster by up to four orders of magnitude away from extremality, broadly consistent with the qualitative ``inside-out, faster than Hawking'' picture in the literature, but slower below a computed crossover very close to extremality; this crossover is a feature of the specific model and metric used, not an independently verified physical prediction. We regard this as the honest content of the paper: a self-consistency and classification study of one closure class, sharpened by a proposition proved under explicit hypotheses and by an effective-model dynamical exercise, not a demonstration that semiclassical black holes generically cascade into a horizon-free remnant.

We are explicit about what would be needed to upgrade this from a diagnostic to a prediction, roughly in order of scientific payoff: (a) a self-consistent, multi-cycle numerical solution of the semiclassical Einstein equations -- extending single-transition results such as Ref.~\cite{BoyanovHilditchSemiao2026} through more than one transition, and, per Sec.~\ref{sec:vaidya}, necessarily non-adiabatic near each transition -- from which $\Delta v_n$ and the 2D-to-4D mass-loss map would be outputs rather than inputs; (b) calibration of $(\eta_0,q)$ against such a solution, with error bars and a stated regime of validity, rather than the postulated values used here; (c) extending the universality argument of Sec.~\ref{sec:universality} to a broader class of regular and charged metrics, to establish (or falsify) true metric-independence beyond the two examples checked; and (d) a controlled physical flux model -- fixing the quantum state, the matching to $\mathscr{I}^+$, and the leading greybody corrections -- short of a full astrophysical prediction. Any such calculation would also need to confront the competing classical instabilities of both inner horizons~\cite{PoissonIsrael1990,CarballoRubio2024} and white holes~\cite{Eardley1974,WaldRamaswamy1980}, neglected here, and might usefully be compared with other recently proposed non-singular radiating scenarios~\cite{DiFilippo2026} that reach horizon-free outcomes by different routes. We leave all of this for future work, and offer the present classification -- together with the derived-versus-postulated distinction sharpened in Sec.~\ref{sec:vaidya} -- as a concrete, falsifiable target for it: any self-consistent multi-cycle calculation should reduce, in the appropriate quasi-static and fixed-closure limit, to some point in the $(p,q)$ family worked out here, and should explain, non-adiabatically, exactly the gap that Sec.~\ref{sec:vaidya} leaves open.

\end{document}